\documentclass[11pt,letterpaper]{article}

\usepackage{amsmath,amsfonts,amsthm,amssymb,stmaryrd,relsize}
\theoremstyle{plain}

\numberwithin{equation}{section}

\newtheorem{thm}{Theorem}[section]
\newtheorem{lem}[thm]{Lemma}

\newenvironment{exam}
{\begin{flushleft}\textbf{Example}.\enspace}%
{\end{flushleft}}

\allowdisplaybreaks  

\newcommand{\complex}{{\mathbb C}}
\newcommand{\Natural}{{\mathbb N}}
\newcommand{\real}{{\mathbb R}}

\newcommand{\tbullet}{\raise .4ex\hbox{\tiny$\bullet$}} 

\newcommand{\ahat}{\widehat{A}}
\newcommand{\bhat}{\widehat{B}}
\newcommand{\ehat}{\widehat{E}}
\newcommand{\phat}{\widehat{P}}
\newcommand{\fhat}{\widehat{f}}
\newcommand{\phihat}{\widehat{\phi}}

\newcommand{\rmtr}{\mathrm{tr\,}}
\newcommand{\ityes}{\textit{yes}}
\newcommand{\itno}{\textit{no}}

\newcommand{\escript}{\mathcal{E}}
\newcommand{\lscript}{\mathcal{L}}
\newcommand{\oscript}{\mathcal{O}}
\newcommand{\pscript}{\mathcal{P}}
\newcommand{\sscript}{\mathcal{S}}

\newcommand{\ab}[1]{\left|#1\right|}
\newcommand{\doubleab}[1]{\left|\left|#1\right|\right|}
\newcommand{\brac}[1]{\left\{#1\right\}}
\newcommand{\paren}[1]{\left(#1\right)}
\newcommand{\sqbrac}[1]{\left[#1\right]}
\newcommand{\elbows}[1]{{\left\langle#1\right\rangle}}
\newcommand{\ket}[1]{{\left|#1\right>}}
\newcommand{\bra}[1]{{\left<#1\right|}}

\begin{document}

\title{CONDITIONED OBSERVABLES\\ IN QUANTUM MECHANICS}
\author{Stan Gudder\\ Department of Mathematics\\
University of Denver\\ Denver, Colorado 80208\\
sgudder@du.edu}
\date{}
\maketitle

\begin{abstract}
This paper presents some of the basic properties of conditioned observables in finite-dimensional quantum mechanics. We begin by defining the sequential product of quantum effects and use this to define the sequential product of two observables. The sequential product is then employed to construct the conditioned observable relative to another observable. We then show that conditioning preserves mixtures and post-process of observables. We consider conditioning among three observables and a complement of an observable. Corresponding to an observable, we define an observable operator in a natural way and show that this mapping also preserves mixtures and post-processing. Finally, we present a method of defining conditioning in terms of self-adjoint operators instead of observables. Although this technique is related to our previous method it is not equivalent.
\end{abstract}

\section{Introduction}  
Various studies in quantum mechanics are based on the results of a measurement conditioned on the value of a previous measurement. For example, one might want to know the position of a particle when the particle is in a given energy state. Although the conditioning of observables seems to be a useful concept there does not appear to be any systematic investigations concerning it. This article does not develop any deep or penetrating results. Instead, it presents an introduction to a theory of conditioned observables. Also, we restrict attention to finite-dimensional quantum mechanics. Although this is a strong restriction, it includes the framework of quantum computation and information theory \cite{hz12,nc00}. These are important topics that have attracted great attention in the recent literature.

We begin with the study of quantum effects.These correspond to simple experiments with only two values or outcomes. These values are usually denoted by \ityes-\itno (or 1\,-\,0). A general effect may be imprecise or fuzzy while a precise effect is called sharp. If $a$ and
$b$ are effects we define their sequential product $a\circ b$ which is the effect that describes the experiment in which $a$ is measured first and then $b$ is measured second. Because of quantum interference, $a$ can interfere with the measurement of $b$, while $b$ cannot interfere with the measurement of $a$. We also call $a\circ b$, the effect $b$ conditioned on the effect $a$ and write $(b\mid a)=a\circ b$. Upon introducing the concept of a state we can also define a corresponding conditional probability.

Now a general observable $A$ may have many possible outcomes $x_1,\ldots ,x_n$. If $a_x$ is the effect that occurs when $A$ has outcome $x$ we can think of $A$ as a set of effects $A=\!\brac{a_x\colon x=x_i, i=\!1,2,\ldots ,n}$. If
$B=\!\brac{b_y\colon y=y_j, j=\!1,2,\ldots ,m}$ is another observable, we shall show in Section~3 how to combine the effects $(b_y\mid a_x)$ to form an observable $(B\mid A)$ that describes $B$ conditioned on $A$. We show that $(B\mid A)$ has a simple form when $A$ and $B$ are sharp observables. We also consider multiple conditionings
$\paren{(B\mid A)\mid C}$ and $\paren{B\mid (A\mid C)}$.

There are two important ways of combining observables called mixtures and post-processing \cite{fghl18,ogwa17}. Section~3 shows that conditioning preserves both of these combination methods. Corresponding to an observable $A$, we define a self-adjoint operator $\ahat$ called the observable operator. The operator $\ahat$ describes $A$ in various ways and we show that
${}^\wedge$ preserves mixtures and post-processing in Section~4. Section~5 discusses a complement of an observable.

Finally, Section~6 considers conditioning from a different point of view. Instead of describing a measurable quantity by an observable, we can describe it by a certain self-adjoint operator. Although this viewpoint is related to our previous work, it is not equivalent to it.

\section{Quantum Effects}  
Let $\lscript (H)$ be the set of linear operators on a finite-dimensional complex Hilbert space $H$. We also denote the set of self-adjoint operators on $H$ by $\lscript _S(H)$ and the zero and unit operators by $0, I$ respectively. For
$S,T\in\lscript (H)$ we write $S\le T$ if $\elbows{\phi ,S\phi}\le\elbows{\phi ,T\phi}$ for all $\phi\in H$. We define the set of \textit{effects} by
\begin{equation*}
\escript (H)=\brac{a\in\lscript (H)\colon 0\le a\le I}
\end{equation*}

An effect $a$ is said to \textit{occur} when a \ityes-\itno\ experiment for $a$ has the value \ityes \cite{dp00,fb94,hz12}. It is well-known that $\escript (H)\subseteq\lscript _S(H)$ \cite{fb94,hz12}. For $a\in\escript (H)$, we call $a'=I-a\in\escript (H)$ the \textit{complement} of $a$ and view $a'$ as the effect that occurs when the previous \ityes-\itno\  experiment has the value \itno. Clearly, $0,I\in\escript (H)$ and $0$ corresponds to the experiment that never occurs (is always \itno) and $I$ responds to the experiment that always occurs (is always \ityes). We denote the set of projections on $H$ by
$\pscript (H)$. It is clear that $\pscript (H)\subseteq\escript (H)$ and we call elements of $\pscript (H)$
\textit{sharp effects} \cite{gud98}. A one-dimensional projection $P_\phi =\ket{\phi}\bra{\phi}$ where $\doubleab{\phi}=1$ is \textit{atomic}. If $\phi\in H$, $\phi\ne 0$ we write $\phihat =\phi\big/\doubleab{\phi}$. We then have
\begin{equation*}
P_{\phihat}=\tfrac{1}{\doubleab{\phi}^2}\,\ket{\phi}\bra{\phi}
\end{equation*}

An effect $\rho\in\escript (H)$ is a \textit{partial state} if the trace $\rmtr (\rho )\le 1$ and $\rho$ is a \textit{state} if
$\rmtr (\rho )=1$. We denote the set of states by $\sscript (H)$. If $\rho\in\sscript (H)$, $a\in\escript (H)$ we call
$E_\rho (a)=\rmtr (\rho a)$ the \textit{probability that} $a$ \textit{occurs} in the state $\rho$. Of course,
$0\le E_\rho (a)\le 1$. If $P_\phi$ is atomic, then $P_\phi\in\sscript (H)$ and we call $P_\phi$ (and $\phi$) a
\textit{pure state}. We then write
\begin{equation*}
E_\phi (a)=E_{P_\phi}(a)=\rmtr (P_\phi a)=\elbows{\phi ,a\phi}
\end{equation*}
If $\phi$ and $\psi$ are pure states, we call $\ab{\elbows{\phi ,\psi}}^2$ the \textit{transition probability} from $\phi$ to
$\psi$.

We denote the unique positive square root of $a\in\escript (H)$ by $a^{1/2}$. For $a,b\in\escript (H)$, their \textit{sequential product} is the effect $a\circ b=a^{1/2}ba^{1/2}$ \cite{gg02,gl08,gn01}. We interpret $a\circ b$ as the effect that results from first measuring $a$ and that $a\circ b=b\circ a$ if and only if $ab=ba$ where $ab$ is the usual operator product \cite{gn01}. This is interpreted as saying that $a$ and $b$ do not interfere if and only if $a$ and $b$ commute. We also call $a\circ b$ the effect $b$ \textit{conditioned on the effect} $a$ and write $(b\mid a)=a\circ b$. For short, we sometimes call $(b\mid a)$ the effect $b$ \textit{given} $a$. We have that $(a\mid a)=a^2$ and $a$ is sharp if and only if $(a\mid a)=a$.

Notice that if $b_1,b_2,b_1+b_2\in\escript (H)$, then $(b_1+b_2\mid a)=(b_1\mid a)+(b_2\mid a)$. In particular,
$\escript (H)$ is convex and if $\lambda _i\ge 0$ with $\sum\lambda _i=1$, then
\begin{equation*}
\paren{\sum\lambda _ib_i\mid a}=\sum\lambda _i(b_i,a)
\end{equation*}
so $b\mapsto (b\mid a)$ is a convex function. Of course, $a\mapsto (b\mid a)$ is not convex in general. Also, for every
$\lambda\in\sqbrac{0,1}\subseteq\real$ we have
\begin{equation*}
(b\mid\lambda a)=(\lambda b\mid a)=\lambda (b\mid a)
\end{equation*}
Moreover,
\begin{equation*}
\rmtr\sqbrac{(b\mid a)}=\rmtr (ba)=\rmtr (ab)=\rmtr\sqbrac{(a\mid b)}
\end{equation*}
If $\rho\in\sscript (H)$, $a\in\escript (H)$, since $\rho\circ a\le\rho$ we have that
\begin{equation*}
\rmtr\sqbrac{(\rho\mid a)}=\rmtr (a\circ\rho )=\rmtr (\rho\circ a)\le\rmtr (\rho )\le 1
\end{equation*}
Hence, $(\rho\mid a)$ is a partial state. For $b\in\escript (H)$ we obtain
\begin{align*}
E_\rho\sqbrac{(b\mid a)}&=\rmtr\sqbrac{\rho (b\mid a)}=\rmtr\sqbrac{\rho\,a\circ b}=\rmtr\sqbrac{(a\circ\rho )b}\\
  &=\rmtr\sqbrac{(\rho\mid a)b}
\end{align*}
We interpret $\rmtr\sqbrac{(\rho\mid a)b}$ as the probability that $b$ occurs for the partial state $(\rho\mid a)$. If
$E_\rho (a)=\rmtr (\rho a)\ne 0$ we can form the state $(\rho\mid a)/\rmtr (\rho a)$. Then as a function of $b$
\begin{equation}                
\label{eq21}
\ehat _\rho\sqbrac{(b\mid a)}=\frac{E_\rho\sqbrac{(b\mid a)}}{E_\rho (a)}
\end{equation}
becomes a probability measure on $\escript (H)$ and we call \eqref{eq21} the \textit{conditional probability of} $b$
\textit{given} $a$.

We now examine some specific examples of $(b\mid a)$. The simplest case is when $a=P_\phi$ is atomic. We then obtain
\begin{equation*}
(b\mid P_\phi )=P_\phi\circ b=\ket{\phi}\bra{\phi}\,b\,\ket{\phi}\bra{\phi}=\elbows{\phi ,b\phi}P_\phi
\end{equation*}
Hence, $(b\mid P_\phi )$ is $P_\phi$ attenuated by the probability of $b$ in the state $\phi$. If $b=P_\phi$ is atomic and $a^{1/2}\phi\ne 0$ we have that
\begin{align*}
(P_\phi\mid a)=a\circ P_\phi&=a^{1/2}\ket{\phi}\bra{\phi}a^{1/2}=\ket{a^{1/2}\phi}\bra{a^{1/2}\phi}\\
  &=\doubleab{a^{1/2}\phi}^2P_{(a^{1/2}\phi )^\wedge}=\elbows{\phi ,a\phi}P_{(a^{1/2}\phi )^\wedge}
\end{align*}
If $a=P_\phi ,b=P_\psi$ are both atomic, we obtain
\begin{equation*}
(P_\psi\mid P_\phi )=P_\phi\circ P_\psi =\elbows{\phi ,P_\psi\phi}P_\phi =\ab{\elbows{\phi ,\psi}}^2P_\phi
\end{equation*}
where $\ab{\elbows{\phi ,\psi}}^2$ is the transition probability from $\phi$ to $\psi$.

More generally, let $P\in\pscript (H)$ so $P$ is a sharp effect. We can then write $P=\sum P_{\phi _i}$ where
$\phi _i$ are mutually orthogonal. We then have
\begin{align*}
(P\mid a)&=a\circ P=\sum a\circ P_{\phi _i}=\sum\elbows{\phi_i,a\phi _i}P_{(a^{1/2}\phi _i)^\wedge}
\intertext{and}
(b\mid P)&=P\circ b=PbP=\sum _{i,j}P_{\phi _i}bP_{\phi _j}
     =\sum _{i,j}\ket{\phi _i}\bra{\phi _i}\,b\,\ket{\phi _j}\bra{\phi _j}\\
     &=\sum _{i,j}\elbows{\phi _i,b\phi _j}\ket{\phi _i}\bra{\phi _j}
\end{align*}

\section{Observables}  
For a finite set $\Omega _A$, an \textit{observable with value-space} $\Omega _A$ is a subset
$A=\brac{a_x\colon x\in\Omega _A}$ of $\escript (H)$ such that $\sum _{x\in\Omega _A}a_x=I$. We write $a_x$ as the effect that occurs when $A$ has the value $x$. The condition $\sum a_x=I$ ensures that $A$ has one of the values
$x\in\Omega _A$. Observables are also called finite \textit{positive operator-valued measures} \cite{hz12,nc00}. If an observable $A$ has only one value, then $A=\brac{I}$ so $A$ is called \textit{trivial}. If $A$ has two values, say \ityes\ and \itno\ then $A=\brac{a,a'}$ where $a\in\escript (H)$ and $a$ is the effect that $A$ has value \ityes, while $a'$ is the effect that $A$ has value \itno .

If $a_x\in\pscript (H)$ for all $x\in\Omega _A$, we call $A$ a \textit{sharp} observable. In this case we have for all
$y\in\Omega _a$ that
\begin{equation*}
a_y+a_y\circ\sum _{x\ne y}a_x=a_y\circ\sum _{x\in\Omega _A}A_x=a_y
\end{equation*}
Hence, $\sum_{x\ne y}a_y\circ a_x=0$ which implies that $a_y\circ a_x=0$ whenever $x\ne y$. We conclude that
$a_ya_x=a_xa_y=0$ whenever $x\ne y$. Hence, the effects for a sharp observable commute and are mutually
orthogonal which makes them much simpler than unsharp observables. If the effects $a_x$, $x\in\Omega _A$, are atoms, we say that the observable $A$ is \textit{atomic}. In this case, $a_x=P_{\phi _x}$ where
$\brac{\phi _x\colon x\in\Omega _A}$ is an orthonormal basis for $H$. In general, if $\Omega _A\subseteq\real$ and
$\rho\in\sscript (H)$, we define the \textit{expectation of} $A$ \textit{in the state} $\rho$ by 
\begin{equation*}
E_\rho (A)=\sum xE_\rho (a_x)=\sum x\ \rmtr (\rho a_x)=\rmtr (\rho\sum xa_x)
\end{equation*}
Notice that $\ahat =\sum xa_x$ is a self-adjoint operator that we call the \textit{observable operator} for $A$. This operator has the same expectations as $A$ for every state $\rho\in\sscript (H)$.

Let $A,B$ be observables with $A=\brac{a_x\colon x\in\Omega _a}$ and $B=\brac{b_y\colon y\in\Omega _B}$. We define their \textit{sequential product} $A\circ B$ to have value-space $\Omega _A\times\Omega _B$ and
\begin{equation*}
A\circ B=\brac{a_x\circ b_y\colon (x,y)\in\Omega _A\times\Omega _B}
\end{equation*}
To show that $A\circ B$ is indeed an observable, we have that
\begin{equation*}
\sum _{(x,y)}a_x\circ b_y=\sum _xa_x\circ\paren{\sum _yb_y}=\sum _xa_x\circ I=\sum _xa_x=I
\end{equation*}
The \textit{left-marginal} of $A\circ B$ consists of the effects
\begin{equation*}
\sum _ya_x\circ b_y=a_x\circ\sum _yb_y=a_x\circ I=a_x
\end{equation*}
so the left-marginal of $A\circ B$ is just $A$. In a similar way the \textit{right-marginal} of $A\circ B$ consists of the effects
\begin{equation*}
\sum _xa_x\circ b_y=\sum _xa_x^{1/2}b_ya_x^{1/2}
\end{equation*}
As before, the right-marginal of $A\circ B$ is an observable but it need not equal $B$ and we denote it by $(B\mid A)$. We thus have that $\Omega _{(B\mid A)}=\Omega _B$ and
\begin{equation*}
(B\mid A)=\brac{\sum _xa_x\circ b_y\colon y\in\Omega _B}=\brac{\sum _x(b_y\mid a_x)\colon y\in\Omega _B}
\end{equation*}
We call $(B\mid A)$ the \textit{observable} $B$ \textit{conditioned on the observable} $A$. For short, we call $(B\mid A)$ the observable $B$ \textit{given} $A$. We denote the effects in $(B\mid A)$ by
\begin{equation*}
(B\mid A)_y=\sum _xa_x\circ b_y=\sum _x(b_y\mid a_x)
\end{equation*}
We use $\oscript (H)$ for the set of observables on $H$.

If $\rho\in\sscript (H)$ and $A\in\oscript (H)$ we define the state $\rho$ \textit{conditioned on} $A$ by
\begin{equation*}
(\rho\mid A)=\sum _xa_x^{1/2}\rho a_x^{1/2}
\end{equation*}
Note that $(\rho\mid A)\in\sscript (H)$ because
\begin{align*}
\rmtr (\rho\mid A)&=\sum _x\rmtr (a_x^{1/2}\rho a_x^{1/2})=\sum _x\rmtr (a_x\rho )\\
  &=\rmtr\paren{\sum _xa_x\rho}=\rmtr (\rho )=1 
\end{align*}

The next result gives a duality between states and observables.

\begin{lem}    
\label{lem31}
If $A,B\in\oscript (H)$ with $\Omega _B\subseteq\real$ and $\rho\in\sscript (H)$, then
\begin{equation*}
E_\rho (B\mid A)=E_{(\rho\mid A)}(B)
\end{equation*}
\end{lem}
\begin{proof}
We have that
\begin{align*}
E_\rho (B\mid A)&=\sum _yy\ \rmtr\sqbrac{\rho (B\mid A)_y}=\sum _yy\ \rmtr (\rho\sum _xa_x\circ b_y)\\
  &=\sum _{x,y}y\ \rmtr\sqbrac{\rho (a_x\circ b_y)}=\sum _{x,y}y\ \rmtr (\rho a_x^{1/2}b_ya_x^{1/2})\\
  &=\sum _{x,y}y\ \rmtr (a_x^{1/2}\rho a_x^{1/2}b_y)\\
  &=\sum _y\rmtr\sqbrac{\paren{\sum _xa_x^{1/2}\rho a_x^{1/2}}b_y}=\sum _yy\ \rmtr\sqbrac{(\rho\mid A)b_y}\\
  &=E_{(\rho\mid A)}(B)\qedhere
\end{align*}
\end{proof}

We now consider $A\circ B$ and $(B\mid A)$ for special cases $A,B\in\oscript (H)$. If $A$ and $B$ are atomic with $a_x=P_{\phi _x}$, $b_y=P_{\psi _y}$, $x\in\Omega _A$, $y\in\Omega _B$ we have that
\begin{equation*}
A\circ B=\brac{\ab{\elbows{\phi _x,\psi _y}}^2P_{\phi _x}\colon x\in\Omega _A,y\in\Omega _B}
\end{equation*}
It follows that
\begin{equation*}
(B\mid A)_y=\sum _x\ab{\elbows{\phi _x,\psi _y}}^2P_{\phi _x}
\end{equation*}
If $A$ is atomic and $B\in\oscript (H)$ is arbitrary, we have
\begin{align*}
A\circ B&=\brac{\elbows{\phi _x,b_y\phi _x}P_{\phi _x}\colon x\in\Omega _A,y\in\Omega _B}\\
\intertext{and}
(B\mid A)_y&=\sum _x\elbows{\phi _x,b_y\phi _x}P_{\phi _x}
\end{align*}
If $A\in\oscript (H)$ is arbitrary and $B$ is atomic, we have
\begin{align*}
A\circ B&=\brac{\elbows{\psi _y,a_x\psi _y}P_{(a_x^{1/2}\psi _y)^\wedge}\colon x\in\Omega _A,y\in\Omega _B}\\
\intertext{and}
(B\mid A)_y&=\sum _x\elbows{\psi _y,a_x\psi _y}P_{(a_x^{1/2}\psi _y)^\wedge}
\end{align*}

Let $B^{(i)}\in\oscript (H)$, $i=1,2,\ldots ,n$, with the same value-space $\Omega$ where
\begin{equation*}
B^{(i)}=\brac{b_y^{(i)}\colon y\in\Omega}
\end{equation*}
For $\lambda _i\in\sqbrac{0,1}$, $i=1,2,\ldots ,n$, with $\sum\lambda _i=1$ we define the \textit{mixture}
$\sum\lambda _iB^{(i)}\in\oscript (H)$ by
\begin{equation*}
\sum _{i=1}^n\lambda _iB^{(i)}=\brac{\sum _{i=1}^n\lambda _ib_y^{(i)}\colon y\in\Omega}
\end{equation*}
It is easy to check that $\sum\lambda _iB^{(i)}$ is indeed on observable. Mixtures are an important way of combining observables and have been well-studied \cite{fghl18,ogwa17}. It is convenient to use the notation
$\paren{\sum\limits _{i=1}^n\lambda _iB^{(i)}}_y=\sum\limits _{i=1}^n\lambda _ib_y^{(i)}$.
\medskip

Let $\Omega _A$, $\Omega _B$ be value-spaces and let $\nu =\sqbrac{\nu _{xy}}$, $x\in\Omega _A$, $y\in\Omega _B$ be a matrix. We call $\nu$ a \textit{stochastic} matrix if $\nu _{xy}\in\sqbrac{0,1}\subseteq\real$ and
$\sum _{y\in\Omega _B}\nu _{xy}=1$ for all $x\in\Omega _A$. The matrix $\nu$ is called a \textit{classical channel} and
$\nu _{xy}$ gives the probability of a transition from $x$ to $y$ \cite{fghl18,ogwa17}. The condition
$\sum _{y\in\Omega _B}\nu _{xy}=1$ means that $x$ makes a transition to some $y\in\Omega _B$ with probability one. Now let $A\in\oscript (H)$ with $A=\brac{a_x\colon x\in\Omega _A}$ and let $\nu$ be a classical channel from $\Omega _A$ to
$\Omega _B$. Define $B\in\oscript (H)$ by $B=\brac{b_y\colon y\in\Omega _B}$ where
$b_y=\sum _{x\in\Omega _A}\nu _{xy}a_x$. Now the value-space of $B$ is $\Omega _B$ and $B$ is indeed an observable because
\begin{equation*}
\sum _yb_y=\sum _y\sum _x\nu _{xy}a_x=\sum _x\sum _y\nu _{xy}a_x=\sum _xa_x=I
\end{equation*}
We use the notation $B=\nu\tbullet A$ and call $B$ a \textit{post-processing} of $A$ \cite{fghl18,ogwa17}. The next result shows that conditioning preserves mixtures and post-processing.

\begin{thm}    
\label{thm32}
{\rm{(i)}}\enspace $A\circ\sum\lambda _iB^{(i)}=\sum\lambda _iA\circ B^{(i)}$.
{\rm{(ii)}}\enspace $\paren{\sum\lambda _iB^{(i)}\mid A}=\sum\lambda _i(B^{(i)}\mid A)$.
{\rm{(iii)}}\enspace If $C=\brac{c_z\colon z\in\Omega _C}$ is an observable, then $(\nu\tbullet A\mid C)=\nu\tbullet (A\mid C)$.
\end{thm}
\begin{proof}
(i)\enspace For any $x\in\Omega _A$ and $y\in\Omega$ we have that
\begin{align*}
\paren{A\circ\sum\lambda _iB^{(i)}}_{(x,y)}&=a_x\circ\paren{\sum\lambda _iB^{(i)}}_y
  =a_x\circ\sum\lambda _ib_y^{(i)}=\sum\lambda _ia_x\circ b_y^{(i)}\\
  &=\sum\lambda _i(A\circ B^{(i)})_{(x,y)}=\paren{\sum\lambda _iA\circ B^{(i)}}_{(x,y)}
\end{align*}
The result now follows.
(ii)\enspace This follows from (i).
(iii)\enspace For all $y\in\Omega$ we have that
\begin{align*}
(\nu\tbullet A\mid C)_y&=\sum _zc_z\circ (\nu\tbullet A)_y=\sum _zc_z\circ\paren{\sum _x\nu _{xy}a_x}
  =\sum _x\nu _{xy}\sum _zc_z\circ a_x\\
  &=\sum _x\nu _{xy}(A\mid C)_x=\sqbrac{\nu\tbullet (A\mid C)}_y
\end{align*}
The result follows.
\end{proof}

We now briefly discuss multiple conditioning. Letting $A,B,C\in\oscript (H)$ we can form the \textit{biconditional}
$\paren{(B\mid A)\mid C}$ in which $C$ is measured first, $A$ is measured second and $B$ is measured last. By definition, we have that
\begin{align*}
\paren{(B\mid A)\mid C}&=\brac{\sum _zc_z\circ (B\mid A)_y\colon y\in\Omega _B}\\
  &=\brac{\sum _zc_z\circ\paren{\sum _xa_x\circ b_y}\colon y\in\Omega _B}\\
  &=\brac{\sum _{z,x}c_z\circ (a_x\circ b_y)\colon y\in\Omega _B}\\
  &=\brac{\sum _{z,y}c_z^{1/2}a_x^{1/2}b_ya_z^{1/2}c_z^{1/2}\colon y\in\Omega _B})
\end{align*}
In particular, if $A=\brac{P_{\alpha _x}}$, $C=\brac{P_{\beta _z}}$ are atomic, we have that
\begin{equation*}
\paren{(B\mid A)\mid C}
  =\brac{\sum _{z,x}\ab{\elbows{\beta _z,\alpha _x}}^2\elbows{\alpha _x,b_y\alpha _x}P_{\beta _z}\colon y\in\Omega _B}
\end{equation*}
Because of nonassociativity, the biconditional is different than
\begin{align*}
\paren{B\mid (A\mid C)}&=\brac{\sum _x(A\mid C)_x\circ b_y\colon y\in\Omega _B}\\
  &=\brac{\sum _x\paren{\sum _zc_z\circ a_x}\circ b_y\colon y\in\Omega _B}
\end{align*}
which cannot be simplified further even if $A,C$ are atomic.

\section{Observable Operators}  
We now consider the observable operator $\ahat =\sum xa_x$ where $A=\brac{a_x\colon x\in\Omega _A}$ and
$\Omega _A\subseteq\real$. In general $\ahat\in\lscript _S(H)$ is not unique. If $A$ is atomic, then $\ahat$ is unique, the values $x\in\Omega _A$ are the eigenvalues of $A$. and $a_x$ is the projection for the corresponding eigenvector. If $f$ is a real-valued function $f\colon\real\to\real$, we define $\fhat (\ahat\,)=\sum f(x)a_x$. The reason we use the notation $\fhat$ is because $\fhat (\ahat\,)$ is not the usual function of an operator. For example, if $f(x)=x^2$, then
\begin{equation*}
\fhat (\ahat\,)=\sum x^2a_x\ne (\ahat\,)^2=f(\ahat\,)
\end{equation*}
If $A$ happens to be sharp, then we do have $\fhat (\ahat\,)=f(\ahat\,)$. In general, $\ahat$ determines $A$ because for any $a_x\in A$ there exists a polynomial $p_x$ such that $a_x=p_x(\ahat\,)$.

If $\nu$ is a classical channel from $\Omega _A$ to $\Omega _B$, we define the function $f_\nu\colon\Omega _A\to\real$ by $f_\nu (x)=\sum _{y\in\Omega _B}y\nu _{xy}$. If we have another channel $\mu$ from $\Omega _B$ to $\Omega _C$, then the matrix product $\nu\mu$ is a classical channel from $\Omega _A$ to $\Omega _C$. Indeed, we have that
\begin{equation*}
\sum _z(\nu\mu )_{xz}=\sum _z\paren{\sum _y\nu _{xy}\mu _{yz}}=\sum _y\nu _{xy}\sum _z\mu _{yz}=\sum _y\nu _{xy}=1
\end{equation*}
so $\nu\mu$ is stochastic.

\begin{lem}    
\label{lem41}
If $\nu$ and $\mu$ are classical channels as above and $A=\brac{a_x\colon x\in\Omega _A}$ is an observable, then
\begin{equation*}
\mu\tbullet (\nu\tbullet A)=(\nu\mu )\tbullet A
\end{equation*}
\end{lem}
\begin{proof}
For all $z\in C$ we have that
\begin{align*}
\sqbrac{\mu\tbullet (\nu\tbullet A)}_z&=\sum _y\mu _{yz}(\nu\tbullet A)_y=\sum _y\mu _{yz}\sum _x\nu _{xy}a_x
  =\sum _x\sum _y\nu _{xy}\mu _{yz}a_x\\
  &=\sum _x(\nu\mu )_{xz}a_x=\sqbrac{(\nu\mu )\tbullet A}_z
\end{align*}
The result follows
\end{proof}

The next result shows that ${}^\wedge$ preserves post-processing and mixtures.

\begin{thm}    
\label{thm42}
{\rm{(i)}}\enspace Using the above notation, we have that
\begin{equation*}
(\nu\tbullet A)^\wedge =\fhat _\nu (\ahat\,)
\end{equation*}
{\rm{(ii)}}\enspace $\sqbrac{\mu\tbullet (\nu\tbullet A)}^\wedge =\fhat _{\nu\mu}(\ahat\,).\quad$
{\rm{(iii)}}\enspace If $\sum\lambda _iB^{(i)}$ is a mixture of the observables $B^{(i)}$, then
\begin{equation*}
\sqbrac{\lambda _iB^{(i)}}^\wedge =\sum _i\lambda _i\sqbrac{B^{(i)}}^\wedge
\end{equation*}
\end{thm}
\begin{proof}
(i)\enspace Since
\begin{align*}
(\nu\tbullet A)^\wedge&=\sum _yy(\nu\tbullet A)_y=\sum _yy\sum _x\nu _{xy}a_x
  =\sum _x\paren{\sum _yy\nu _{xy}}a_x\\
  &=\sum _xf_\nu (x)a_x=\fhat _\nu (\ahat\,)
\end{align*}
The result follows.\newline
(ii)\enspace The result follows from Part~(i) and Lemma~\ref{lem41}.\newline
(iii)\enspace Since
\begin{equation*}
\sqbrac{\sum\lambda _iB^{(i)}}^\wedge =\sum _xx\sqbrac{\sum _i\lambda _iB^{(i)}}_x
  =\sum _i\lambda _i\sum _xxB_x^{(i)}=\sum\lambda _i\sqbrac{B^{(i)}}^\wedge
\end{equation*}
the result follows.
\end{proof}

If $A=\brac{a_{(x,y)}\colon (x,y)\in\Omega _A}$ is an observable with value-space $\Omega _A\subseteq\real ^2$, we define the \textit{observable operator} of $A$ by $\ahat =\sum _{x,y}xya_{(x,y)}$. If $a\in\escript (H)$ and $T\in\lscript (H)$ we use the notation
\begin{equation*}
(T\mid a)=a^{1/2}Ta^{1/2}=a\circ T
\end{equation*}

\begin{thm}    
\label{thm43}
If $A=\brac{a_x\colon x\in\Omega _A}$ and $B=\brac{b_y\colon y\in\Omega _B}$ are real-valued observables then
{\rm{(i)}}\enspace $(B\mid A)^\wedge =\sum _x(\bhat\mid a_x)$ and
{\rm{(ii)}}\enspace $(A\circ B)^\wedge =\sum _xx(\bhat\mid a_x)=\sum _xx(a_x\circ\bhat\,)$.
\end{thm}
\begin{proof}
(i)\enspace The result follows from
\begin{align*}
(B\mid A)^\wedge&=\sum _yy(B\mid A)_y=\sum _yy\sum _xa_x\circ b_y=\sum _x\paren{a_x\circ\sum _yyb_y}\\
  &=\sum _x(a_x\circ\bhat\,)=\sum _x(\bhat\mid a_x)
\end{align*}
(ii)\enspace The result follows from
\begin{align*}
(A\circ B)^\wedge&=\sum _{x,y}xy(A\circ B)_{(x,y)}=\sum _{x,y}xya_x\circ b_y=\sum _{x,y}xya_x^{1/2}b_ya_x^{1/2}\\
    &=\sum _xxa_x^{1/2}\sum _yyb_ya_x^{1/2}=\sum _xxa_x^{1/2}\bhat a_x^{1/2}=\sum _xx(\bhat\mid a_x)\\
   &=\sum _xx(a_x\circ\bhat\,)\qedhere
\end{align*}
\end{proof}

\begin{exam}  
The simplest example is the qubit Hilbert space $H=\complex^2$ and dichotomic (two-valued) atomic observables 
$A=\brac{P_{\phi _1},P_{\phi _2}}$, $B=\brac{P_{\psi _1},P_{\psi _2}}$ where $\brac{\phi _1,\phi _2}$, $\brac{\psi _1,\psi _2}$ are orthonormal bases for $H$. The sequential product observable becomes
\begin{align*}
A\circ B
  &=\brac{P_{\phi _1}\circ P_{\psi _1}, P_{\phi _1}\circ P_{\psi _2},P_{\phi _2}\circ P_{\psi _1},P_{\phi _2}\circ P_{\psi _2}}\\
  &=\brac{\ab{\elbows{\phi _1,\psi _1}}^2P_{\phi _1},\ab{\elbows{\phi _1,\psi _2}}^2P_{\phi _1},%
  \ab{\elbows{\phi _2,\psi _1}}^2P_{\phi _2},\ab{\elbows{\phi _2,\psi _2}}^2P_{\phi _2}}
\end{align*}
Letting $\Omega _A=\brac{x_1,x_2}$, $\Omega _B=\brac{y_1,y_2}$, $B$ conditioned on $A$ is the observable
\begin{equation*}
(B\mid A)_{y_i}=(A\circ B)_{(x_1,y_i)}+(A\circ B)_{(x_2,y_i)}
  =\ab{\elbows{\phi _1,\psi _i}}^2P_{\phi _1}+\ab{\elbows{\phi _2,\psi _i}}^2P_{\phi _2}
\end{equation*}
for $i=1,2$. If $\Omega _A,\Omega _B\subseteq\real$, the observable operators become
\begin{equation*}
\ahat =x_1P_{\phi _1}+x_2P_{\phi _2},\quad\bhat =y_1P_{\psi _1}+y_2P_{\psi _2}
\end{equation*}
Applying Theorem~\ref{thm43}(i) we obtain
\begin{align*}
(B\mid A)^\wedge&=\elbows{\phi _1,\bhat\phi _1}P_{\phi _1}+\elbows{\phi _2,\bhat\phi _2}P_{\phi _2}\\
  &=\sqbrac{y_1\ab{\elbows{\phi _1,\psi _1}}^2+y_2\ab{\elbows{\phi _1,\psi _2}}^2}P_{\phi _1}\\
  &\qquad +\sqbrac{y_1\ab{\elbows{\phi _2,\psi _1}}^2+y_2\ab{\elbows{\phi _2,\psi _2}}^2}P_{\phi _2}\\
  &=\sqbrac{y_2+(y_1-y_2)\ab{\elbows{\phi _1,\psi _1}}^2}P_{\phi _1}
  +\sqbrac{y_2+(y_1-y_2)\ab{\elbows{\phi _2,\psi _1}}^2}P_{\phi _2}\\
  &=\sqbrac{y_2+(y_1-y_2)\ab{\elbows{\phi _1,\psi _1}}^2}P_{\phi _1}
  +\sqbrac{y_1-(y_2-y_1)\ab{\elbows{\phi _1,\psi _1}}^2}P_{\phi _2}
\end{align*}
Moreover, applying Theorem~\ref{thm43}(ii) gives
\begin{align*}
(A\mid B)^\wedge&=x_1\elbows{\phi _1,\bhat\phi _1}P_{\phi _1}+x_2\elbows{\phi _2,\bhat\phi _2}P_{\phi _2}\\
  &=\sqbrac{x_1y_1\ab{\elbows{\phi _1,\psi _1}}^2+x_1y_2\ab{\elbows{\phi _1\psi _2}}^2}P_{\phi _1}\\
  &\qquad +\sqbrac{x_2y_1\ab{\elbows{\phi _2,\psi _1}}^2+x_2y_2\ab{\elbows{\phi _2,\psi _2}}^2}P_{\phi _2}\\
  &=x_1\sqbrac{y_2+(y_1-y_2)\ab{\elbows{\phi _1,\psi _1}}^2}P_{\phi _1}\\
   &\qquad +x_2\sqbrac{y_2+(y_1-y_2)\ab{\elbows{\phi _2,\psi _1}}^2}P_{\phi _2}\\
   &=x_1\sqbrac{y_2+(y_1-y_2)\ab{\elbows{\phi _1,\psi _1}}^2}P_{\phi _1}\\
   &\qquad +x_2\sqbrac{y_1+(y_2-y_1)\ab{\elbows{\phi _1,\psi _1}}^2}P_{\phi _2}\hskip 6pc\qed
\end{align*}
\end{exam}

\section{An Observable Complement}  
Let $A=\brac{a_x\colon x\in\Omega _A}$ be an observable. We call $A$ an $n$-\textit{observable} if $\ab{\Omega _A}=n$ and $a_x\ne 0$ for all $x\in\Omega _A$. We define the $n$-observable
\begin{equation*}
I_A=\brac{\tfrac{1}{n}\,I_x\colon x\in\Omega _A}
\end{equation*}
where $I_x=I$ for all $x\in\Omega _A$. It is easy to check that $(I_A\mid B)=I_A$ and $(B\mid I_A)=B$ for every
$B\in\oscript (H)$. If $\lambda\in\sqbrac{0,1}$ we call $\lambda I_A+(1-\lambda )A$ the observable $A$ with
\textit{noise content} $\lambda$ \cite{hz12}. We define the \textit{complement} of an $n$-observable $A$ by
\begin{equation*}
A'=\brac{\tfrac{1}{n-1}\,a'_x\colon x\in\Omega _A}
\end{equation*}
The reader can easily verify that $A'$ is indeed an observable.

\begin{lem}    
\label{lem51}
For an $n$-observable $A$ we have that $A'=A$ if and only if $A=I_A$.
\end{lem}
\begin{proof}
For sufficiency we have that
\begin{align*}
I'_A&=\brac{\tfrac{1}{n-1}\,\paren{\tfrac{1}{n}\,I_x}'\colon x\in\Omega _A}
  =\brac{\tfrac{1}{n-1}\,\paren{I-\tfrac{1}{n}\,I_x}\colon x\in\Omega _A}\\
  &=\brac{\tfrac{1}{n-1}\,\paren{1-\tfrac{1}{n}}I_x\colon x\in\Omega _A}=\brac{\tfrac{1}{n}\,I_x\colon x\in\Omega _A}=I_A
\end{align*}
For necessity, if $A'=A$, we obtain for all $x\in\Omega _A$ that
\begin{equation*}
a_x=\tfrac{1}{n-1}\,a'_x=\tfrac{1}{n-1}\,(I-a_x)=\tfrac{1}{n-1}\,I-\tfrac{1}{n-1}\,a_x
\end{equation*}
This implies that $a_x=\tfrac{1}{n}\,I$. Hence, $A=I_A$.
\end{proof}

The next result shows that complementation preserves conditioning and mixtures.

\begin{thm}    
\label{thm52}
{\rm{(i)}}\enspace $(B\mid A)'=(B'\mid A)$ for all $A,B\in\oscript (H)$.
{\rm{(ii)}}\enspace If $\lambda _i\in\sqbrac{0,1}$, $i=1,2,\ldots ,m$, $\sum\lambda _i=1$ and $A_i\in\oscript (H)$, $i=1,2,\ldots ,m$, with the same value-spaces $\Omega$, then
\begin{equation*}
\paren{\sum\lambda _iA_i}'=\sum\lambda _iA'_i
\end{equation*}
\end{thm}
\begin{proof}
(i)\enspace The sequential product $A\circ B'$ becomes
\begin{align*}
A\circ B'&=\brac{a_x\circ\paren{\tfrac{1}{n-1}\,b'_y}\colon (x,y)\in\Omega _A\times\Omega _B}\\
  &=\brac{\tfrac{1}{n-1}\,a_x\circ b'_y\colon (x,y)\in\Omega _A\times\Omega _B}\\
  &=\brac{\tfrac{1}{n-1}\,a_x\circ (I-b_y)\colon (x,y)\in\Omega _A\times\Omega _B}\\
  &=\brac{\tfrac{1}{n-1}\,(a_x-a_x\circ b_y)\colon (x,y)\in\Omega _A\times\Omega _B}
\end{align*}
Hence,
\begin{align*}
(B'\mid A)&=\brac{\tfrac{1}{n-1}\,\sum _x(a_x-a_x\circ b_y)\colon y\in\Omega _B}\\
  &=\brac{\tfrac{1}{n-1}\paren{I-\sum _xa_x\circ b_y}\colon y\in\Omega _B}\\
  &=\brac{\tfrac{1}{n-1}\paren{\sum _xa_x\circ b_y}'\colon y\in\Omega _B}=(B\mid A)'
\end{align*}
(ii)\enspace For $A_i=\brac{a_{ix}\colon x\in\Omega}$ we have that
\begin{align*}
\paren{\sum\lambda _iA_i}'&=\brac{\tfrac{1}{n-1}\paren{\sum\lambda _ia_{ix}}'\colon x\in\Omega}
  =\brac{\tfrac{1}{n-1}\paren{I-\sum\lambda _ia_{ix}}\colon x\in\Omega}\\
  &=\brac{\tfrac{1}{n-1}\paren{\sum\lambda _iI-\sum\lambda _ia_{ix}}\colon x\in\Omega}\\
  &=\brac{\tfrac{1}{n-1}\,\sum\lambda _i(I-a_{ix})\colon x\in\Omega}
    =\brac{\tfrac{1}{n-1}\,\sum\lambda _ia'_{ix}\colon x\in\Omega}\\
  &=\brac{\sum\lambda _i\,\tfrac{1}{n-1}\,a'_{ix}\colon x\in\Omega}=\sum\lambda _iA'_i\qedhere
\end{align*}
\end{proof}
We say that a stochastic matrix $\nu$ is \textit{bistochastic} if $\sum _x\nu _{xy}=1$ for all $y$. Although complementation need not preserve post-processing we have the following result.

\begin{lem}    
\label{lem53}
$(\nu\tbullet A)'=\nu\tbullet A'$ if and only if $\nu$ is bistochastic.
\end{lem}
\begin{proof}
We have that
\begin{equation*}
(\nu\tbullet A)'=\brac{\tfrac{1}{n-1}\paren{\sum _x\nu _{xy}a_x}'\colon y\in\Omega _B}=
  \brac{\tfrac{1}{n-1}\paren{I-\sum _x\nu _{xy}a_x}\colon y\in\Omega _B}
\end{equation*}
Moreover,
\begin{align*}
\nu\tbullet A'&=\brac{\tfrac{1}{n-1}\,\sum _x\nu _{xy}a'_x\colon y\in\Omega _B}
  =\brac{\tfrac{1}{n-1}\,\sum _x\nu _{xy}(I-a_x)\colon y\in\Omega _B}\\
  &=\brac{\tfrac{1}{n-1}\paren{\sum _x\nu _{xy}I-\sum _x\nu _{xy}a_x}\colon y\in\Omega _B}
\end{align*}
These two expressions agree if and only if $\sum _x\nu _{xy}=1$ for all $y$.
\end{proof}

It is of interest to iterate the complementation operation to obtain\newline $A',A'',A''',A^{iv},A^v,\ldots\,$.

\begin{thm}    
\label{thm54}
Let $A$ be an $n$-observable. If $m$ is even, then
\begin{equation}                
\label{eq51}
A^m=\sqbrac{1-\tfrac{1}{(n-1)^m}}I_A+\tfrac{1}{(n-1)^m}\,A
\end{equation}
and if $m$ is odd, then
\begin{equation}                
\label{eq52}
A^m=\sqbrac{1-\tfrac{1}{(n-1)^{m-1}}}I_A+\tfrac{1}{(n-1)^{m-1}}\,A'
\end{equation}
\end{thm}
\begin{proof}
The statement clearly holds for $m=1$. To show it holds for $m=2$ we have that
\begin{align} 
\label{eq53}
A''&=\brac{\tfrac{1}{n-1}\paren{\tfrac{1}{n-1}\,a'_x}'\colon x\in\Omega _A}
  =\brac{\tfrac{1}{n-1}\paren{I-\tfrac{1}{n-1}\,a'_x}\colon x\in\Omega _A}\notag\\
  &=\brac{\tfrac{1}{n-1}\sqbrac{I-\tfrac{1}{n-1}(I-a_x)}\colon x\in\Omega _A}\notag\\
  &=\brac{\sqbrac{\tfrac{1}{n-1}-\tfrac{1}{(n-1)^2}}I+\tfrac{1}{(n-1)^2}\,a_x\colon x\in\Omega _A}\notag\\
  &=\brac{\tfrac{(n-2)n}{(n-1)^2}\,\tfrac{1}{n}\,I_x+\tfrac{1}{(n-1)^2}\,a_x\colon x\in\Omega _A}\notag\\
  &=\sqbrac{1-\tfrac{1}{(n-1)^2}}I_A+\tfrac{1}{(n-1)^2}\,A
\end{align}
Proceeding by induction, suppose the result holds for the integer $m$. If $m$ is even, then \eqref{eq51} holds. Applying Lemma~\ref{lem51} and Theorem~\ref{thm52}(ii) we conclude that
\begin{equation*}
A^{m+1}=\sqbrac{1-\tfrac{1}{(n-1)^m}}I_A+\tfrac{1}{(n-1)^m}\,A'
\end{equation*}
which is \eqref{eq52} with $m$ replaced by $m+1$. Hence, the result holds for $m+1$. If $m$ is odd, then \eqref{eq52} holds. Again, by Lemma~\ref{lem51} and Theorem~\ref{thm52}(ii) we obtain
\begin{equation*}
A^{m+1}=\sqbrac{1-\tfrac{1}{(n-1)^{m-1}}}I_A+\tfrac{1}{(n-1)^{m-1}}\,A''
\end{equation*}
Applying \eqref{eq53} we conclude
\begin{align*}
A^{m+1}&=\sqbrac{1-\tfrac{1}{(n-1)^{m+1}}}I_A+\tfrac{1}{(n-1)^{m+1}}\sqbrac{1-\tfrac{1}{(n-1)^2}}I_A
  +\tfrac{1}{(n-1)^{m+1}}\,A\\
  &=\sqbrac{1-\tfrac{1}{(n-1)^{m+1}}}I_A+\tfrac{1}{(n-1)^{m+1}}\,A
\end{align*}
which is \eqref{eq51} with $m$ replaced by $m+1$. Hence, the result again holds for $m+1$. It follows by induction that the result holds for all $m\in\Natural$.
\end{proof}

We conclude from Theorem~\ref{thm54} that if $m$ is even, then $A^m$ is the observable $A$ with noise content
$\sqbrac{1-\tfrac{1}{(n-1)^m}}$ and if $m$ is odd, then $A^m$ is the observable $A'$ with noise content
$\sqbrac{1-\tfrac{1}{(n-1)^{m-1}}}$. The dichotomic $(n=2)$ case is an exception and we then have that $A^m=A$ when $m$  is even and $A^m=A'$ when $m$ is odd. Notice that $A'$ is a special case of a post-processing of $A$. In fact, $A'=\nu\tbullet A$ where for all $x,y\in\Omega _A$ we have that
\begin{equation*}
\nu _{xy}=\begin{cases}\tfrac{1}{n-1}&\hbox{if }x\ne y\\\noalign{\smallskip}0&\hbox{if }x=y\end{cases}
\end{equation*}

\section{A Different Viewpoint}  
We now consider conditioning from another viewpoint. Besides observables, measurable quantities are frequently represented by self-adjoint operators. For $T\in\lscript _S(H)$, the corresponding \textit{spectral observable} is given by the unique sharp observable $P=\brac{P_x}$ where $T=\sum xP_x$, $P_x\in\pscript (H)$, $x\in\real$. In this case, the $x$ are the distinct eigenvalues of $T$. Notice that $P$ is a real-valued observable and $\phat =\sum xP_x=T$ so our concepts are consistent. Let $S\in\lscript _S(H)$, with $S=\sum yQ_y$, $Q_y\in\pscript (H)$, $y\in\real$ so $Q=\brac{Q_y}$ is the spectral observable for $S$. Letting $\Omega _T$, $\Omega _S$ be the sets of eigenvalues for $T$ and $S$, respectively, we have that
\begin{equation*}
Q\circ P=\brac{Q_y\circ P_x\colon x\in\Omega _T,y\in\Omega _S}
\end{equation*}
and $(P\mid Q)_x=\sum _{y\in\Omega _S}Q_y\circ P_x$. We then define the operator
$(T\mid S)\in\lscript (H)$ by$(T\mid S)=(P\mid Q)^\wedge$ and call $(T\mid S)$ the operator $T$
\textit{conditioned on the operator} $S$. We then have that
\begin{align}                
\label{eq61}
(T\mid S)&=\sum _xx(P\mid Q)_x=\sum _xx\sum _yQ_y\circ P_x=\sum _xx\sum _yQ_yP_xQ_y\notag\\
  &=\sum _yQ_y\paren{\sum _xxP_x}Q_y=\sum _yQ_yTQ_y
\end{align}
It is interesting to note that $(T\mid S)$ depends on $T$ and $Q_y$, $y\in\Omega _S$, but not on the particular values of $y$.

\begin{lem}    
\label{lem61}
We have that $(T\mid S)=T$ if and only if $ST=TS$.
\end{lem}
\begin{proof}
If $ST=TS$ then it is well-known that $Q_yT=TQ_y$ for all $y\in\Omega _S$. Applying \eqref{eq61} gives
\begin{equation*}
(T\mid S)=\sum _yQ_yT=T
\end{equation*}
Conversely, suppose that $(T\mid S)=T$. Applying \eqref{eq61} again, we obtain $T=\sum _yQ_yTQ_y$. It follows that
\begin{equation*}
Q_yT=Q_yTQ_y=TQ_y
\end{equation*}
for all $y\in\Omega _S$ so that $ST=TS$.
\end{proof}

Notice that $T\mapsto (T\mid S)$ is a real linear function.

\begin{thm}    
\label{thm62}
{\rm{(i)}}\enspace If $T\ge 0$, then $(T\mid S)\ge 0$.
{\rm{(ii)}}\enspace $\rmtr\sqbrac{(T\mid S)}=\rmtr (T)$.
{\rm{(iii)}}\enspace If $\rho\in\sscript (H)$, then $(\rho\mid S)\in\sscript (H)$ and
\begin{equation*}
\rmtr\sqbrac{\rho (T\mid S)}=\rmtr\sqbrac{(\rho\mid S)T}
\end{equation*}
\end{thm}
\begin{proof}
(i)\enspace Assume that $T\ge 0$ and $\phi\in H$. Applying \eqref{eq61} gives 
\begin{align*}
\elbows{\phi ,(T\mid S)\phi}&=\elbows{\phi ,\sum _yQ_yTQ_y\phi}=\sum _y\elbows{\phi ,Q_yTQ_y\phi}\\
  &=\sum _y\elbows{Q_y\phi ,TQ_y\phi}\ge 0
\end{align*}
Hence, $(T\mid S)\ge 0$.
(ii)\enspace Again, applying \eqref{eq61} gives
\begin{align*}
\rmtr\sqbrac{(T\mid S)}&=\rmtr\paren{\sum _yQ_yTQ_y}=\sum _y\rmtr (Q_yTQ_y)\\
  &=\sum _y\rmtr (Q_yT)=\rmtr\paren{\sum Q_yT}=\rmtr (T)
\end{align*}
(iii)\enspace If $\rho\in\sscript (H)$, it follows from (i) and (ii) that $(\rho\mid S)\in\sscript (H)$. Moreover, it follows from \eqref{eq61} that 
\begin{align*}
\rmtr\sqbrac{\rho (T\mid S)}&=\rmtr\sqbrac{\rho\sum _yQ_yTQ_y}=\rmtr\sqbrac{\sum _yQ_y\rho Q_yT}\\
  &=\rmtr\sqbrac{(\rho\mid S)T}\qedhere
\end{align*}
\end{proof}

When $Q$ is atomic with $Q=\brac{P_{\psi _y}}$, then \eqref{eq61} becomes
\begin{equation*}
(T\mid S)=\sum _y\elbows{\psi _y,T\psi _y}P_{\psi _y}=\sum _{x,y}x\elbows{\psi _y,P_x\psi _y}P_{\psi _y}
\end{equation*}
and when $P$ is atomic with $P=\brac{P_{\phi _x}}$, then \eqref{eq61} becomes
\begin{equation*}
(T\mid S)=\sum _{x,y}x\elbows{\phi _x,Q_y\phi _x}P_{(Q_y\phi _x)^\wedge}
\end{equation*}
When $P$ and $Q$ are both atomic as above, then \eqref{eq61} gives
\begin{equation*}
(T\mid S)=\sum _{x,y}x\ab{\elbows{\psi _y,\phi _x}}^2P_{\psi _y}
\end{equation*}
where $\ab{\elbows{\psi _y,\phi _x}}^2$ is the transition probability from $\phi _x$ to $\psi _y$.

Although this technique is related to our previous method, it is not equivalent because the observables are sharp.

\end{document}